\documentclass{article}
\usepackage{amssymb}
\usepackage{graphicx}
\usepackage{amsmath}
\usepackage{tikz}

\newtheorem{theorem}{Theorem} [section]

\newtheorem{corollary}[theorem]{Corollary}

\newtheorem{definition}[theorem]{Definition}
\newtheorem{example}[theorem]{Example}

\newtheorem{lemma}[theorem]{Lemma}

\newtheorem{proposition}[theorem]{Proposition}

\newenvironment{proof}[1][Proof]{\textbf{#1.} }{\ \rule{0.5em}{0.5em}}

\begin{document}

\author{Luis A. Guardiola\thanks{%
Departamento de Fundamentos del An\'{a}lisis Econ\'{o}mico, Universidad de
Alicante, Alicante 03071, Spain. E-mail: luis.guardiola@ua.es} , Ana Meca%
\thanks{
Operations Research Center. Universidad Miguel Hern\'{a}ndez, Edificio
Torretamarit. Avda. de la Universidad s.n. 03202 Elche (Alicante), Spain.
E-mail: ana.meca@umh.es} and Justo Puerto\thanks{%
Facultad de Matem\'{a}ticas, Universidad de Sevilla, 41012 Sevilla, SPAIN.
e-mail: puerto@us.es} \thanks{%
Corresponding author.}}
\title{Quid Pro Quo allocations\\
in Production-Inventory games \thanks{
The research of the authors is partially supported by Ministry of Econom{\'\i%
}a and Competitividad/FEDER grants numbers: {\color{blue}%
MTM2016-74983-C02-01, PGC2018-097965-B-I00.} }}
\maketitle

\begin{abstract} 
The concept of Owen point, introduced in Guardiola et al. (2009), is an appealing solution concept  that for Production-Inventory games (PI-games) always belongs to their core. The Owen point allows all the players in the game to operate at minimum cost but it does not take into account the cost reduction induced by essential players over their followers (fans). Thus,  it may be seen as an altruistic allocation for essential players what can be criticized. The aim this paper is two-fold: to study the structure and complexity of the core of PI-games and to introduce new core allocations for PI-games improving the weaknesses of the Owen point. Regarding the first goal, we advance further on the analysis of PI-games and we analyze its core structure and algorithmic complexity. Specifically, we prove that the number of extreme points of the core of PI-games is exponential on the number of players. On the other hand, we propose and characterize  a new core-allocation, the Omega point, which compensates the essential players for their role on reducing the costs of their fans. Moreover, we define another solution concept, the Quid Pro Quo set (QPQ-set) of allocations, which is based on the Owen and Omega points. Among all the allocations in this set, we emphasize what we call the \textit{Solomonic} QPQ allocation and we provide some necessary conditions for the coincidence of that allocation with the Shapley value and the Nucleolus.
\medskip


\end{abstract}

\newpage

\section{Introduction}

Guardiola et al. (2009) introduced Production-Inventory games (henceforth
PI-games) as a new class of totally balanced combinatorial optimization
games. That paper proposed the so-called
\textit{Owen point}  core-allocation that allows all players to operate at minimum cost at the price of not compensating  essential users by the cost reduction that they induce over the remaining players (fans). This allocation has proven to be rather appealing and in another paper, Guardiola et al. (2008) analyze its properties and propose three axiomatic
characterizations for the Owen point.
These papers also contribute to a better knowledge of the core of PI-games. Nevertheless, it was missing a deeper analysis of its complexity. Specifically speaking the two following aspects were not considered: testing core membership and the extreme points structure of the core of these games. Complexity issues in cooperative game theory raise important questions only partially answered for particular classes of games. The {core of} any convex game is the convex hull of its marginal vectors (Shapley 1971), and the same property holds true for those games satisfying the Co-Ma property which
include, among others, assignment and information games, see Hamers et al.
(2002) and Kuipers (1993) respectively. It is also well-known that the core
of assignment games coincide with the allocations induced by dual solutions
and it is a complete lattice with only two extreme points, see Sotomayor
(2003). Also, for transportation games, which constitute an extension of the
assignment games, some results about the relationship between the core and
the allocations induced by dual solutions are provided by S\'{a}nchez-Soriano et al. (2001).
Moreover, Perea et al. (2012) study cooperation situations in linear production
problems. In particular, that paper proposes a new solution concept called EOwen set as an improvement of the Owen set that contains at least one allocation that assigns a strictly positive payoff to players necessary for optimal production plans.

For minimum cost spanning tree games, flow games, linear
production games, cooperative facility location games or  min-coloring
games among others, testing whether a given allocation is in the core is an
NP-complete problem (see Faigle et al. (1997), Fang et al. (2002), Goemans and
Skutella (2004) and Deng et al. (1999), respectively). On the other hand,
there are some classes of games for which testing core membership is polynomially solvable as for instance for routing games, see Derks and Kuipers (1997), $s-t$ connectivity games, $r$-arborescence games, max matching games, min vertex cover games, min edge cover
games or max independent set games, see e.g., Deng et al. (1999). However, for many other classes of cooperative games answering that question is still open, as it is the case of PI-games.

In this paper we investigate the structure of the core of PI-games by determining its algorithmic complexity. Our contribution is to prove that testing core membership is an NP-complete problem and moreover that the number of extreme points of the core of PI-games is exponential on the number of players. Specifically, we characterize an exponential size subset of them. In addition, we look for alternative cost allocations improving the fairness properties of the Owen  point in that they recognize the role of the essential players on reducing the costs of the remaining players.

To present our results the rest of the paper is organized as follows. We start by introducing some preliminary
concepts in section \ref{preli}. In section \ref{structure} we prove that testing core membership of PI-games in an NP-complete problem, and we analyze the
core structure of PI-games. We define what we call the \textit{extreme functions}, which help us to prove that the core of a PI-game, in general, has an exponential number
of extreme points. In section \ref{structure copy 2} we introduce a new
core-allocation for PI-games, the Omega point, and provide an axiomatic characterization. Finally, in section \ref{QPQ set} we define the set of Quid Pro Quo
allocations (henceforth, QPQ allocations). Every QPQ allocation is a convex
combination of the Owen and the Omega point. We focus then on the equally
weighted QPQ allocation, the \textit{Solomonic} allocation, and we provide some
necessary conditions for the coincidence of the latter with the Shapley value
and the Nucleolus.

\section{\label{preli}Preliminaries}

A cost game with transferable utility (henceforth TU cost game) is a pair $%
(N,c)$, where $N=\left\{ 1,2,...,n\right\} $ is the finite set of players,
and the characteristic function $c:\mathcal{P}(N)\rightarrow \mathbb{R}$, is
defined over $\mathcal{P}(N)$ the set of nonempty coalitions of $N$. By
agreement, it always satisfies $c(\varnothing )=0.$ For all $S\subseteq N$,
we denote by $\left\vert S\right\vert $ the cardinal of the set $S$.

A distribution of the costs of the grand coalition, usually called
cost-sharing vector, is a vector $x\in \mathbb{R}^{N}$. For every coalition $%
S\subseteq N$ we denote by $x_{S}:=\sum_{i\in S}x_{i}$ the cost-sharing of
coalition $S$ (where $x_{\varnothing }=0).$ The core of a TU cost game
consists of those cost-sharing vectors $x$ which allocate the cost of the grand
coalition $N$ in such a way that no coalition $S$ has incentives to leave $N$ because $x(S)$ is smaller than the original cost of $S$, $c(S)$. Formally, the core of $(N,c)$
is given by $Core(N,c)=\left\{ x\in \mathbb{R}^{n}\left/ x_{N}=c(N)\text{
and }x_{S}\leq c(S)\text{ for all }S\subset N\right. \right\} .$ In the
following, core-allocations will be cost-sharing vectors belonging to the
core. A cost game $(N,c)$ is balanced if and only if has a nonempty core
(see Bondareva 1963 or Shapley 1967). Shapley and Shubik (1969) describe
totally balanced games as those games whose subgames are also balanced;
i.e., the core of every subgame is nonempty. A cost game $(N,c)$ is concave
if for all $i\in N$ and all $S,T\subseteq N$ such that $S\subseteq T\subset
N $ with $i\in S,$ then $c(S)-c(S\setminus \{i\})\geq c(T)-c(T\setminus
\{i\}). $

The Shapley value (Shapley, 1953) is a linear function on the class of all
TU games and for a cost game $(N,c)$ it is defined as $\phi (N,c)=(\phi
_{i}(N,c))_{i\in N}$ where for all $i\in N$

\begin{equation*}
\phi _{i}(N,c)=\sum\limits_{S\subseteq N\backslash \{i\}}\frac{s!(n-s-1)!}{n!%
}\cdot \left[ c(S\cup \{i\})-c(S)\right] .
\end{equation*}

The Nucleolus $\eta (N,c)$ (Schmeidler, 1969) is the allocation that
lexicographically minimizes the vector of excesses. It is well-known that the Nucleolus is a
core-allocation provided that the core is nonempty.

Let $Q$ be a bounded convex polyhedron in $\mathbb{R}^{n}$. We say that $%
x\in Q$ is an extreme point if $y,z\in Q$ and $x=\frac{1}{2}y+\frac{1}{2}z$
imply $y=z$. From now on, we denote, respectively, by $Ext\left( Q\right)$ and by $\partial
\left( Q\right)$ the set of extreme points and the boundary of the set of $Q$. Moreover, for the sake of readability, we use $e_{i}$ to refer to the $i$-th element of the canonical basis of $\mathbb{R}^{n}$ and $val(P)$ stands for the optimal value of the mathematical programming problem $P$.

It is well-known that $x\in Ext\left( Q\right) $ if
and only if $x$ satisfies as equalities at least $n$ linearly independent constraints of those defining $Q$. Since the core is a bounded
convex polyhedron, it has a finite number of extreme points. Moreover, the
core is a convex set. Therefore, characterizing the extreme core-allocations
is important to know its intrinsic structure.

From now on, and for the sake of readability, we follow the same notation as
Guardiola et al. (2009) to describe Production-Inventory situations
(henceforth: PI-situations) and PI-games. Consider first a situation with
several agents facing each one a Production-Inventory problem. Then, they
decide to cooperate to reduce costs. Here the cooperation is considered as
sharing technologies in production, inventory carrying and backlogged
demand. We mean that if a group of agents agree on cooperation then at each
period they will produce and pay inventory carrying and backlogged demand at
the cheapest costs among the members of the coalition. This situation is
called a PI-situation.

Formally, let $U$\ be an infinite set, the universe of players. A
PI-situation is a 3-tuple $(N,D,\Re )$ where $N\subset U$ is a finite set of
players $\left( \left\vert N\right\vert =n\right) $ and  $D$ an integer
matrix of demands with $D=[d^{1},\ldots ,d^{n}]^{\prime }$, $d^{i}=[d_{1}^{i},%
\ldots ,d_{T}^{i}]\geq 0$, $d_{t}^{i}$ is the demand of the player $i$ during
period $t\in T$ and $T$ is the planning horizon. In addition, $\Re =(H|B|P)$ is a cost matrix, so that $\quad H=[h^{1},\ldots ,h^{n}]^{\prime
}$, $B=[b^{1},\ldots ,b^{n}]^{\prime }$ and $P=[p^{1},\ldots
,p^{n}]^{\prime };$ where  $h^{i}=[h_{1}^{i},\ldots ,h_{T}^{i}]\geq 0$, $h_{t}^{i}$ is the unit inventory carrying costs of the player $i$ in period $t$, $%
b^{i}=[b_{1}^{i},\ldots ,b_{T}^{i}]\geq 0$, $b_{t}^{i}$ is the unit backlogging
carrying costs of the player $i$ in period $t$, and $p^{i}=[p_{1}^{i},\ldots
,p_{T}^{i}]\geq 0$, $p_{t}^{i}$ the unit production costs of the player $i$ in
period $t$, for $t=1,\ldots ,T$. The decision variables of the model, which
are required to be integer quantities, are the production during period $t$ (%
$q_{t}$), the inventory at hand at the end of period $t$ ($I_{t}$), and the
backlogged demand at the end of period $t$ ($E_{t}$). We denote by $\Upsilon $ the set of PI-situations $(N,D,\Re )$ defined over $U$, being $n\geq 1,T\geq 1$ and $D$ an integer
matrix.\medskip

Now given a PI-situation $(N,D,\Re )$, we can associate the corresponding TU
cost game $(N,c)$ with the following characteristic function $c$: $%
c(\varnothing )=0$ and for any $S\subseteq N,c(S)=val(PI(S))$, where $
PI(S) $ is given by
\begin{eqnarray*}
(PI(S))\quad &\min
&\sum_{t=1}^{T}(p_{t}^{S}q_{t}+h_{t}^{S}I_{t}+b_{t}^{S}E_{t}) \\
&\mbox{s.t.}&I_{0}=I_{T}=E_{0}=E_{T}=0, \\
&&I_{t}-E_{t}=I_{t-1}-E_{t-1}+q_{t}-d_{t}^{S},\quad t=1,\ldots ,T, \\
&&q_{t},\;I_{t},\;E_{t},\text{ non-negative, integer, }t=1,\ldots ,T;
\end{eqnarray*}%
with
\begin{equation*}
p_{t}^{S}=\min_{i\in S}\{p_{t}^{i}\},\;h_{t}^{S}=\min_{i\in
S}\{h_{t}^{i}\},\;b_{t}^{S}=\min_{i\in
S}\{b_{t}^{i}\},\;d_{t}^{S}=\sum_{i\in S}d_{t}^{i}.
\end{equation*}

Every TU cost game defined as above is called a Production-Inventory game.
Guardiola et al. (2009) points out that the problem $ PI(S)$
has integer optimal solutions provided that the demands are integer. We know
that the dual problem of $PI(S)$, for any coalition $S\subseteq N$, is the following
mathematical programming problem,
\begin{eqnarray*}
(DLPI(S))\quad &\max &\sum_{t=1}^{T}d_{t}^{S}y_{t}\mbox{   } \\
&\mbox{s.t.}&y_{t}\leq p_{t}^{S},\qquad \qquad t=1,\ldots ,T, \\
&&y_{t+1}-y_{t}\leq h_{t}^{S},\quad t=1,\ldots ,T-1, \\
&&-y_{t+1}+y_{t}\leq b_{t}^{S},\quad t=1,\ldots ,T-1.
\end{eqnarray*}

Moreover, Guardiola et al. (2009) also proves that an optimal solution of problem $%
DLPI(S)$ is $y_{t}^{\ast }(S)=\min \Big\{p_{t}^{S},\min_{k<t}%
\{p_{k}^{S}+h_{kt}^{S}\},\min_{k>t}\{p_{k}^{S}+b_{tk}^{S}\}\Big\}$, for all $%
t=1,\ldots ,T,$ with
\begin{eqnarray*}
p_{k}^{S} &=&\left\{
\begin{array}{cc}
p_{1}^{S} & \text{if }k<1, \\
p_{T}^{S} & \text{if }k>T,%
\end{array}%
\right. \\
h_{kt}^{S} &=&\sum_{r=k}^{t-1}h_{r}^{S},\quad \mbox{for any
}k<t,t=2,\ldots ,T;h_{k1}^{S}=0,k<1, \\
b_{tk}^{S} &=&\sum_{r=t}^{k-1}b_{r}^{S},\quad \mbox{ for any }%
k>t,\;t=1,\ldots ,T-1;b_{Tk}^{S}=0,k>T.
\end{eqnarray*}

It is important to note that those optimal solutions satisfy a monotonicity
property with respect to coalitions : $y_{t}^{\ast }(S)\geq y_{t}^{\ast }(R)$
for all $S\subseteq R\subseteq N$\ and all $t\in \{1,...,T\}$. Moreover, the
characteristic function of PI-games can be rewritten as follows: for any $%
\varnothing \neq S\subseteq N,c(S)=\sum_{t=1}^{T}d_{t}^{S}y_{t}^{\ast
}\left( S\right) $. \label{p:seis}

PI-games are not concave in general as shown by Example 4.4 in Guardiola et
al. (2009). The allocation $\left( \sum_{t=1}^{T}d_{t}^{i}y_{t}^{\ast
}\left( N\right) \right) _{i\in N}\hspace*{-0.3cm}=Dy^{\ast }(N)$ is called
the Owen point, and it is denoted by $Owen(N,D,\Re )$. At times, for the
sake of simplicity, we use $o$ to refer to the Owen point. That same paper also proves
that the Owen point is a core-allocation which can be reached
through a PMAS (Sprumont, 1990); hence every PI-game is a totally balanced
game. In some situations we will use $c^{(N,D,\Re )}(S)$ instead of $c(S)$,
in order to denote that the game $\left( N,c\right) $ comes from the
situation $(N,D,\Re )$.

 We say that a player $i\in N$ is essential if there exists $t\in
\{1,...,T\}$ with $d_{t}^{N\backslash \{i\}}>0$ such that $y_{t}^{\ast
}(N\backslash \{i\})>y_{t}^{\ast }(N)$. An essential player is the one for
which there exists at least one period in which he is needed by the rest of
players in order to produce a certain demand at a minimum cost. The set of
essential players is denoted by $\mathcal{E}$. Those players not being
essential are called inessential. We can easily check that for each
inessential player $i$,\ $o_{N\setminus\{i\}}=c(N\setminus \{i\})$.
Guardiola et al. (2009) showed that the core of PI-games shrinks to a single
point, the Owen point, just only when all players are inessential for the
PI-situation.

Finally, to conclude this section devoted to preliminaries, we recall the class of PS-games introduced by {Kar et al. (2009)}. A PS-game $(N,c)$ is a TU cost game satisfying that for all player $i\in N, $ there {exists} a real constant $c_{i\text{ }}$ such that $\Delta
_{i}(S)+\Delta _{i}(N\setminus \left( S\cup \{ i \}\right)=c_{i\text{ }}$ for
all $S\subseteq N\setminus \{i\},$ where $\Delta _{i}(S):=c(S\cup
\{i\})-c(S).$ The above mentioned paper proves that, for this class of games, the Shapley value and the Nucleolus coincide; i.e. $\phi (N,c)=\eta (N,c)$.

\section{\label{structure} Extreme points of the core of PI-games}

Guardiola et al. (2009) demostrated that the core of PI-games without
essential players ($\mathcal{E}=\varnothing $) shrinks to a singleton,
the Owen point. However, for those PI-games with essential players ($
\mathcal{E}\neq \varnothing $), the core is large. We focus here on those
PI-games with large cores and study the structure of its core by
analyzing its extreme points. First of all, we remark that testing core membership for PI-games cannot be done in polynomial time. One can adapt the
reduction proposed in Fang et al. (2002) to prove
that checking if an imputation belongs to the core of a PI-game is
an NP-complete decision problem. In spite of that, it is important to know the structure of the core and still very little is known about the extreme points complexity of PI-games. This is the goal of this section.
\smallskip

We begin this analysis by defining the essential player fan set.

Let $(N,D,\Re )$ be a PI-situation with $D$ being an integer matrix ($
(N,D,\Re )\in \Upsilon $), and let $i$ be an essential player. We define the
fan set of $i$\ as follows:
\begin{equation*}
F_{i}:=\{j\in N\backslash \{i\}\left\vert \exists t\in \{1,...,T\}\text{
with }d_{t}^{j}>0\text{ and }y_{t}^{\ast }(N\backslash \{i\})>y_{t}^{\ast
}(N)\}.\right.
\end{equation*}

The fan set of player $i$ consists of all players who need
him to operate at a lower cost. It is always a non-empty set.
Indeed, $F_{i}\neq \varnothing$ since taking $i\in \mathcal{E},$ there exists
$t^{\ast }\in \{1,...,T\}$ such that $y_{t^{\ast }}^{\ast }(N\backslash \{i\})>y_{t^{\ast }}^{\ast }(N)$
and $d_{t^{\ast }}^{N\backslash \{i\}}>0$. In that case, there must be, at least, a player
$j\in N\backslash \{i\}$ such that $d_{t^{\ast }}^{j}>0$ and $y_{t^{\ast }}^{\ast }(N\backslash \{i\})>y_{t^{\ast }}^{\ast }(N)$ .
\smallskip

In addition, you may notice that there is a pairwise relationship among essential players and their fans,
in the sense that the latter are interested in taking on a portion of the costs of the former.
This relationship allows us to introduce the concept of essential-fan pair.

Let $(N,D,\Re )\in \Upsilon $. The essential-fan pair set, denoted by $\mathbb{P}$, is:
\begin{equation*}
\mathbb{P}:=\{(i,j)|i\in \mathcal{E}\text{ and }j\in F_{i}\}.
\end{equation*}

We are now interested in determining the cost that can be transferred within every essential-fan pair with a cost  allocation;  i.e., the maximum portion of the essential player cost that his fan could assume while maintaining cooperation.

Given a essential-fan pair $p=\left( i,j\right) \in \mathbb{P}$ and a allocation $
x\in \mathbb{R}^{n},$ the transferred cost induced by $p$ regarding $x$ is:
\begin{equation*}
\alpha _{p}(x):=\min_{R\in \Delta _{p}}\{c(R)-x_{R}\},
\end{equation*}%
where
\begin{equation*}
\Delta _{(i,j)}:=\{R\subseteq N\backslash \{i\}\text{ such that }j\in R\}.
\end{equation*}

$\alpha _{p}(x)$ can be interpreted as the maximum portion of cost of player $i$ that can be awarded by player $j$ while maintaining the cooperation of the group.  It is worth nothing that if $x\in Core(N,c)$ then $\alpha _{p}(x)\geq 0$.
\smallskip

Next result states that there are always a positive transferred cost within every essential-fan pair with the
Owen point.

\begin{lemma}
\label{o}Let $(N,D,\Re )\in \Upsilon $ and $(N,c)$ be the corresponding
PI-game. Then $\alpha _{p}(o)>0$ for all $p\in \mathbb{P}$.
\end{lemma}

\begin{proof}
As $\mathcal{E}\neq \varnothing $, we can take $i\in \mathcal{E}$ and
therefore $F_{i}\neq \varnothing $. Let $R\subseteq N\setminus \{i{\}}$ such
that $R\cap F_{i}\neq \varnothing $ and let $j\in R\cap F_{i}$. By
definition, there exists $t^{\ast }\in \{1,...,T\}$ such that $y_{t^{\ast
}}^{\ast }(N\backslash \{i\})>y_{t^{\ast }}^{\ast }(N)$ and $d_{t^{\ast
}}^{j}>0$. Then $d_{t^{\ast }}^{R}>0$ and moreover $y_{t^{\ast }}^{\ast
}(N)<y_{t^{\ast }}^{\ast }(N\backslash \{i\})\leq y_{t^{\ast }}^{\ast }(R)$.
Thus, $o_{R}<c(R)$. Hence, $\alpha _{p}(o)=\min_{R\in \Delta
_{p}}\{c(R)-o_{R}\}>0$.
\end{proof}
\smallskip

We introduce now a function that transforms any cost allocation into a new cost allocation in which a fan player charges with the maximum cost of his essential player. That is, for each $p=(i,j)\in \mathbb{P}$, the function $f_{p}$
transforms any allocation $x$ into a new allocation $f_{p}(x),$ in which the fan player $j$ assumes as much cost as possible from his essential player $i.$ It is called the extreme function.

\begin{definition}[extreme function]
Let $(N,D,\Re )\in \Upsilon $ and $(N,c)$ be the corresponding PI-game. For
any $p=(i,j)\in \mathbb{P}$, the extreme function $f_{p}$ is defined by:
\begin{equation*}
f_{p}(x)=x+\wedge _{p}(x),
\end{equation*}%
where $x\in \mathbb{R}^{n}$ and $\wedge _{p}(x)=e_{j}\cdot \alpha
_{p}(x)-e_{i}\cdot \alpha _{p}(x).$
\end{definition}

Let us denote by $\mathbb{P}^{\left\vert \mathbb{P}\right\vert }$ the $%
\left\vert \mathbb{P}\right\vert $-fold cartesian product of the set $%
\mathbb{P}$. We consider now the composition of extreme functions.
For each $\sigma \in \mathbb{P}^{\left\vert \mathbb{P}\right\vert }$ we define
the extreme composite function, $F_{\sigma }$, as the composition of extreme
functions for all the pairs in $\sigma $, that is,
\begin{equation*}
F_{\sigma }(x):=\left( f_{\sigma _{\left\vert \mathbb{P}\right\vert }}\circ
f_{\sigma _{\left\vert \mathbb{P}\right\vert -1}}\circ ...\circ f_{\sigma
_{1}}\right) \left( x\right) .
\end{equation*}

Notice that if $\sigma =\left( p,p,...,p\right) \in \mathbb{P}^{\left\vert
\mathbb{P}\right\vert }$ then $F_{\sigma }(x)=f_{p}(x).$

\begin{example}
\label{example 1} The following table shows a PI-situation with three
periods and three players:

\begin{equation*}
\begin{tabular}{|c|c|c|c||c|c|c||c|c||c|c|}
\hline
& \multicolumn{3}{|c||}{Demand} & \multicolumn{3}{|c||}{Production} &
\multicolumn{2}{|c||}{Inventory} & \multicolumn{2}{|c|}{Backlogging} \\
\hline
P1 & 10 & 10 & 5 & 1 & 2 & 1 & 1 & 1 & 1 & 1 \\ \hline
P2 & 8 & 12 & 6 & 2 & 1 & 1 & 1 & 1 & 1 & 1 \\ \hline
P3 & 6 & 5 & 2 & 3 & 1 & 1 & 1 & 1 & 2 & 2 \\ \hline
\end{tabular}%
\end{equation*}

We can easily check that $c(S)=\displaystyle\sum_{t=1}^{3}p_{t}^{S}d_{t}^{S},$ for all $S\subseteq N$.
Hence, the characteristic function of the corresponding PI-game is given in the following table:

\begin{equation*}
\begin{array}{|c|c|c|c||c|c|c||c|c||c|c||c||}
\hline
& d_{1}^{_{S}} & d_{2}^{_{S}} & d_{3}^{_{S}} & p_{1}^{_{S}} & p_{2}^{_{S}} &
p_{3}^{_{S}} & h_{1}^{_{S}} & h_{2}^{_{S}} & b_{1}^{_{S}} & b_{2}^{_{S}} & c
\\ \hline
\{1\} & 10 & 10 & 5 & 1 & 2 & 1 & 1 & 1 & 1 & 1 & 35 \\ \hline
\{2\} & 8 & 12 & 6 & 2 & 1 & 1 & 1 & 1 & 1 & 1 & 36 \\ \hline
\{3\} & 6 & 5 & 2 & 3 & 1 & 1 & 1 & 1 & 2 & 2 & 25 \\ \hline
\{1,2\} & 18 & 22 & 11 & 1 & 1 & 1 & 1 & 1 & 1 & 1 & 51 \\ \hline
\{1,3\} & 16 & 15 & 7 & 1 & 1 & 1 & 1 & 1 & 1 & 1 & 38 \\ \hline
\{2,3\} & 14 & 17 & 8 & 2 & 1 & 1 & 1 & 1 & 1 & 1 & 53 \\ \hline
\{1,2,3\} & 24 & 27 & 13 & 1 & 1 & 1 & 1 & 1 & 1 & 1 & 64 \\ \hline
\end{array}%
\end{equation*}

In this example, $y^{\ast }(N)=\left( 1,1,1\right) $\ and $y^{\ast
}(N\setminus \{1\})=\left( 2,1,1\right) $\ are the optimal solution for $%
\left( DLPB(N)\right) $\ and $\left( DLPB(N\setminus \{1\})\right) ,$\
respectively. Then, the Owen point is $o=(25,26,13)$. Moreover, $\mathcal{E}%
=\{1\}$, $F_{1}=\{2,3\}$ and $\mathbb{P}=\{(1,2),(1,3)\}$.

The transferred cost within every essential-fan pair in $\mathbb{P}$ with the
Owen point are,
\begin{eqnarray*}
\alpha _{(1,2)}(o) &=&\min_{R\in \Delta _{(1,2)}}\{c(R)-o_{R}\}=10, \\
\alpha _{(1,3)}(o) &=&\min_{R\in \Delta _{(1,3)}}\{c(R)-o_{R}\}=12.
\end{eqnarray*}

Therefore, the extreme functions are
\begin{eqnarray*}
f_{(1,2)}(o) &=&o+\wedge _{(1,2)}(o)=(15,36,13), \\
f_{(1,3)}(o) &=&o+\wedge _{(1,3)}(o)=(13,26,25).
\end{eqnarray*}

In this case, both the extreme functions, $f_{(1,2)}(o)$, $f_{(1,3)}(o)$, and the Owen point, $o$,  are extreme points of the core.
\end{example}

The previous example shows that the Owen point is an extreme point of the core, and that the extreme functions transform it into other extreme points of the core. We wonder then if this fact occurs in general for any PI-game. First, we find a very interesting property that relates the extreme functions to the core boundary.

\begin{proposition}
Let $(N,D,\Re )\in \Upsilon $ and $(N,c)$ be the corresponding PI-game. For
all $p\in \mathbb{P}$%
\begin{equation*}
f_{p}(Core(N,c))\subseteq \partial {(Core(N,c)).}
\end{equation*}
\end{proposition}

\begin{proof}
Let $p=(i,j)\in \mathbb{P}$ and take $x\in Core(N,c).$ Then $\alpha
_{p}(x)\geq 0$. Applying the extreme function $f_{p}$ at $x$, we have:
\begin{equation*}
f_{p}(x)=(x_{1},...,x_{i-1},x_{i}-\alpha
_{p}(x),x_{i+1},...,x_{j-1},x_{j}+\alpha _{p}(x),x_{j+1},...,x_{n}).
\end{equation*}

To prove that $y:=f_{p}(x)\in Core(N,c)$ we distinguish four possibilities:

\begin{itemize}
\item $i,j\in S.$ Then $y_{S}=x_{S}+\alpha _{p}(x)-\alpha _{p}(x)=x_{S}\leq
c(S)$.

\item $i,j\notin S.$ Then $y_{S}=x_{S}\leq c(S)$.

\item $i\notin S,j\in S.$ Then $y_{S}=x_{S}+\alpha _{p}(x)\leq
x_{S}+c(S)-x_{S}=c(S)$.

\item $i\in S,j\notin S.$ Then $y_{S}=x_{S}-\alpha _{p}(x)\leq x_{S}\leq
c(S) $.
\end{itemize}

Hence, $y\in Core(N,c)$ since $y_{S}\leq c(S)$ for any coalition $S\subseteq
N$. Let us proof now that $y$ belongs to the frontier of the core.

If $\alpha _{p}(x)=0$ then there exists $R\in \Delta _{p}$ such that $%
c(R)=x_{R}$. Since $y$ belongs to the core and satisfies as equality one of
the constraints defining the core, we can conclude that $y\in \partial {%
(Core(N,c))}$.

If $\alpha _{p}(x)>0$ then for all $\lambda \in (0,1),$ $(1-\lambda
)x+\lambda y\in Core(N,c).$ Take $\lambda =1+\epsilon $ with $\epsilon >0$
to have%
\begin{equation*}
(1-\lambda )x+\lambda y=-\epsilon x+(1+\epsilon )(x+\wedge
_{p}(x))=x+(1+\epsilon )\wedge _{p}(x).
\end{equation*}

We can check that if $R^{\ast }\in \Delta _{p}$ is such that $c(R^{\ast
})-x_{R^{\ast }}=\min_{R\in \Delta _{p}}\{c(R)-x_{R}\}$ then $x_{R^{\ast
}}+\alpha _{p}(x)=c(R^{\ast })$, therefore $x+(1+\epsilon )\wedge
_{p}(x)\notin Core(N,c)$. Hence, $y$ is not an interior point.
\end{proof}
\smallskip

It follows straightforward from the above proposition,  that $F_{\sigma
}(Core(N,c))\subseteq$ $\partial {(Core\left( N,c\right) )}$ for all $\sigma
\in \mathbb{P}^{\left\vert \mathbb{P}\right\vert }$.
\smallskip

The main Theorem of this Section provides a partial answer to our previous question
about the transformation of the Owen point into extreme points of the core of PI-games.
It states that for PI-situations with a single essential player, all the different compositions of extreme functions
over the Owen point generate extreme points of the core.

\begin{theorem}
Let $(N,D,\Re )\in \Upsilon $ and $(N,c)$ be the corresponding PI-game. If $%
\mathcal{E}=\{i\}$, then $F_{\sigma }(o)\in Ext\left( Core(N,c)\right) $ for
all $\sigma \in \mathbb{P}^{\left\vert \mathbb{P}\right\vert }$.
\end{theorem}

\begin{proof}
Let $j\in F_{i}$ then the pair $p_{j}=(i,j)\in \mathbb{P}$. $f_{p_{j}}(o)$
is an extreme point if for any $y,z\in Core\left( N,c\right) $ such that

\begin{equation}
f_{p_{j}}(o)=\frac{1}{2}y+\frac{1}{2}z\mbox{ we have that }y=z.  \label{eq1}
\end{equation}

By definition, we know that%
\begin{equation*}
f_{p_{j}}(o)=(o_{1},...,o_{i-1},o_{i}-\alpha
_{p_{j}}(o),o_{i+1},...,o_{j-1},o_{j}+\alpha _{p_{j}}(o),o_{j+1},...,o_{n}).
\end{equation*}

Let us suppose that $z_{k}<o_{k}$ for any $k\neq i,j$ then $z_{N\setminus
\{k\}}>o_{N\setminus \{k\}}=c\left( N\setminus \{k\}\right) .$ However this
is not possible, therefore $y_{k},z_{k}\geq o_{k}$ for all $k\neq i,j$. Now,
apply (\ref{eq1}) to get that $y_{k}=z_{k}=o_{k}$ $\forall k\neq i,j$.
Moreover, $y_{j},z_{j}\leq o_{j}+\alpha _{p_{j}}(o)$ since $\alpha
_{p_{j}}(o)>0$ is the maximum possible increment for $o_{j}$ (see Lemma \ref%
{o}). Then by (\ref{eq1}) we have that $z=y$ and hence $f_{p_{j}}(o)\in
Ext\left( Core(N,c)\right) $.

Now, we consider $p_{l}=(i,l)\in \mathbb{P}$, and apply the corresponding
extreme function for this pair. We have that%
\begin{equation*}
f_{p_{l}}\left( f_{p_{j}}(o)\right) =(o_{1},...,o_{i}-\alpha
_{p_{j}}(o)-\alpha _{p_{l}}(f_{p_{j}}(o)),...,o_{l}+\alpha
_{p_{l}}(f_{p_{j}}(o)),...,o_{j}+\alpha _{p_{j}}(o),...,o_{n}).
\end{equation*}

We distinguish two possibilities:

\begin{enumerate}
\item $\alpha _{p_{j}}(o)$ attains its minimum in a coalition $R^{\ast }$
that contains player $l$. In this case $\alpha _{p_{l}}(f_{p_{j}}(o))=0$,
thus $f_{p_{l}}\left( f_{p_{j}}(o)\right) =f_{p_{j}}(o)$ and by the argument
above $f_{p_{l}}\left( f_{p_{j}}(o)\right) $ is an extreme point of $%
Core(N,c)$.

\item $\alpha _{p_{j}}(o)$ attains its minimum in a coalition $R^{\ast }$
that does not contain player $l$. This case implies that $\alpha
_{p_{l}}(f_{p_{j}}(o))>0$. Take $y,z\in Core(N,c)$ and assume that
\begin{equation}
f_{p_{l}}\left( f_{p_{j}}(o)\right) =\frac{1}{2}y+\frac{1}{2}z.  \label{eq2}
\end{equation}

Using the same argument as above we conclude that $y_{k}=z_{k}=o_{k}$ for
all $k\neq i,j,l.$ Consider now the $j$-th coordinate. Suppose that $%
z_{j}>o_{j}+\alpha _{p_{j}}(o)$. The coalition $R^{\ast }$ does not contain
neither $i$ nor $l$, which implies $c(R^{\ast })=o_{R^{\ast }}+\alpha
_{p_{j}}(o)<z_{R^{\ast }}.$ Since this is a contradiction, it means that $%
z_{j}\leq o_{j}+\alpha _{p_{j}}(o)$ (Notice that the same argument applies
to $y_{j}$\ and thus $y_{j}\leq o_{j}+\alpha _{p_{j}}(o)$). Therefore, by (%
\ref{eq2}) we get that $y_{j}=z_{j}=o_{j}+\alpha _{p_{j}}(o)$.

Next, consider the $l$-th coordinate. Assume that $z_{l}>o_{l}+\alpha
_{p_{l}}(f_{p_{j}}(o))$, and let $S^{\ast }$\ be the coalition where $\alpha
_{p_{l}}(f_{p_{j}}(o))$ attains its minimum, then $c(S)=\sum_{k\in S}\left(
f_{p_{j}}(o)\right) _{k}+\alpha _{p_{l}}(f_{p_{j}}(o))<z_{S}$. Again using
the same argument as in the $j$-th coordinate we conclude that $%
y_{l}=z_{l}=o_{l}+\alpha _{p_{l}}(f_{p_{j}}(o)).$

Finally, we get the same conclusion for the $i$-th coordinate since $%
f_{p_{l}}\left( f_{p_{j}}(o)\right) $ must be efficient. In conclusion $z=y.$
Hence, $f_{p_{l}}\left( f_{p_{j}}(o)\right) \in Ext\left( Core(N,c)\right) $.
Notice that $f_{p_{l}}\left( f_{p_{j}}(o)\right) $ is different from $%
f_{p_{j}}(o)$ since we have assumed that $\alpha _{p_{l}}(f_{p_{j}}(o))>0$.
\end{enumerate}

This construction can be repeated a finite number of times for each $p\in
\mathbb{P}$. Specifically, for any $\sigma \in \mathbb{P}^{\left\vert
\mathbb{P}\right\vert },$ the transformation $F_{\sigma }(o)\in Ext\left(
Core(N,c)\right) $.
\end{proof}

\begin{corollary}
Let $(N,D,\Re )\in \Upsilon $ with $\mathcal{E}=\{i\}$, and $(N,c)$ be the corresponding PI-game.
The Owen point is always an extreme point.
\end{corollary}

\begin{proof}
Take $y_{k},z_{k}\geq o_{k}$ for all $k\in N$, therefore $o=z=y$ and $o\in Ext\left( Core(N,c)\right).$
\end{proof}
\smallskip

At this point we know that PI-games with a single essential player have, at least,
$\left\vert \mathbb{P}\right\vert+1$ extreme points. Next example shows that
the core of a PI-game, in general, cannot be explicitly described
in polynomial time.

\begin{example}
\label{neco} Now we consider a PI-situation with n periods and n players:

\begin{equation*}
\begin{tabular}{|c|c|c|c|c||c|c|c|c||c|c|c|c||c|c|c|c|}
\hline
& \multicolumn{4}{|c||}{Demand} & \multicolumn{4}{|c||}{Production} &
\multicolumn{4}{|c||}{Inventory} & \multicolumn{4}{|c|}{Backlogging} \\
\hline
P1 & 1 & 1 & $\dots $ & 1 & $\frac{1}{n}$ & $\frac{1}{n}$ & $\dots $ & $%
\frac{1}{n}$ & 2 & 2 & $\dots $ & 2 & 2 & 2 & $\dots $ & 2 \\ \hline
P2 & 1 & 1 & $\dots $ & 1 & 1 & 1 & $\dots $ & 1 & 2 & 2 & $\dots $ & 2 & 2
& 2 & $\dots $ & 2 \\ \hline
$\vdots $ & $\vdots $ & $\vdots $ & $\ddots $ & $\vdots $ & $\vdots $ & $%
\vdots $ & $\ddots $ & $\vdots $ & $\vdots $ & $\vdots $ & $\ddots $ & $%
\vdots $ & $\vdots $ & $\vdots $ & $\ddots $ & $\vdots $ \\ \hline
Pn & 1 & 1 & $\dots $ & 1 & 1 & 1 & $\dots $ & 1 & 2 & 2 & $\dots $ & 2 & 2
& 2 & $\dots $ & 2 \\ \hline
\end{tabular}%
\end{equation*}

The corresponding PI-game is given by $c(S)=\displaystyle%
\sum_{t=1}^{n}p_{t}^{S}d_{t}^{S}$, for all $S\subseteq N$. Moreover, it is
easy to see that $\mathcal{E}=\{1\}$ and $F_{1}=\{2,3,...,n\}$. Then, we can
rewrite the characteristic function as follows:

\begin{equation*}
c(S)=\left\{
\begin{array}{cc}
\left\vert S\right\vert & \text{if }1\in S, \\
&  \\
n\cdot \left\vert S\right\vert & \text{if }1\notin S.%
\end{array}%
\right.
\end{equation*}

In this example, the Owen point is $o=(1,1,...,1)$. For all $i\in F_{1}$,
\begin{equation*}
\alpha _{(1,i)}(o)=\min_{R\in \Delta _{(1,i)}}\{c(R)-o_{R}\}=\min_{R\in
\Delta _{(1,i)}}\{n\cdot \left\vert R\right\vert -\left\vert R\right\vert
\}=n-1
\end{equation*}%
then%
\begin{equation*}
f_{(1,i)}(o)=(2-n,1,...,1,\underset{i}{\underbrace{n}},1,...,1).
\end{equation*}

For all $k\neq i,$ $k\in F_{1}$,
\begin{eqnarray*}
\alpha _{(1,k)}(f_{(1,i)}(o)) &=&\min_{R\in \Delta
_{(1,k)}}\{c(R)-\sum_{j\in R}\left( f_{(1,i)}(o)\right) _{j}\} \\
&=&\min_{R\in \Delta _{(1,k)}}\{\underset{\text{if }i\notin R}{\underbrace{%
n\cdot \left\vert R\right\vert -\left\vert R\right\vert }},\underset{\text{%
if }i\in R}{\underbrace{\left( n-1\right) \cdot \left\vert R\right\vert -n+1}%
}\}=n-1,
\end{eqnarray*}%
then%
\begin{equation*}
f_{(1,k)}\left( f_{(1,i)}(o)\right) =(3-2n,1,...,1,\underset{i}{\underbrace{n%
}},1,...,1,\underset{k}{\underbrace{n}},1,...,1).
\end{equation*}

Hence, we have as many extreme points as possible ways to place $"n"$ and $%
"1"$ in $n-1$ positions; i.e. in this example the core has $2^{n-1}+1$
extreme points.
\end{example}

Therefore, we can conclude that the cardinality of the extreme
points is exponential in the number of players. Hence, we cannot explicitly
describe the core of a PI-game in polynomial time.
\smallskip

We propose below an alternative core allocation to the Owen point that recognizes the role
played by essential players on reducing the cost of their fans.

\section{\label{structure copy 2} Omega point}

Guardiola et al. (2009) proposed the Owen point as a natural core allocation for PI-games that
arises when focusing on shadow prices of each period that each player must pay to meet their
demand in that period. It  makes it possible for all players in the joint venture to operate
at minimum cost.  If there is no essential player, the Owen point is the unique core
allocation. However, for those PI-situations with at least one essential player, the Owen point
reveals the altruistic character of them because of it does not take into account the role that these
essential players play in reducing the cost of their fans. As the core of the PI-games with essential
players is large, we are looking for a core allocation that motivates the essential players to continue
in the join venture obtaining a reduction in their demand costs in each period.
\smallskip

Let $(N,D,\Re )$ be a PI-situation with $D$ being an integer matrix ($
(N,D,\Re )\in \Upsilon $), and $\mathcal{E}\neq \varnothing $. Remember
that for all  $i\in\mathcal{E}$, there is a period $t^{\ast }\in \{1,...,T\}$
such that $y_{t^{\ast }}^{\ast }(N\backslash \{i\})>y_{t^{\ast }}^{\ast }(N)$
and there also exists at least one player $j\in N\backslash \{i\}$ such that
$d_{t^{\ast }}^{j}>0$. We denote by $\mathcal{E}^{t}$ and $F^{t}$ the sets of essential players and fans
for every period $t\in \{1,...,T\}$. We note in passing that $\mathcal{E}=\underset{t\in T}{\bigcup }\mathcal{E}^{t}$.
\smallskip

First, we consider the marginal contribution of the shadow prices of a player $i$ to the grand coalition $N$,
that is, $ y_{t}^{\ast }(N\backslash \{i\})-y_{t}^{\ast }(N).$  We then define the cost reduction that a player $i\in N$ can produce in another player $j\in N$ in a period $t$ as follows:
\begin{equation*}
q_{t}(i,j):=\left\{
\begin{array}{ccc}
\left( y_{t}^{\ast }(N\backslash \{i\})-y_{t}^{\ast }(N)\right) \cdot
d_{t}^{j} & if & i\neq j \\
0 & if & i=j\text{ }%
\end{array}%
\right.
\end{equation*}

The reader may notice that $q_{t}(i,j)>0$ only if $i\in \mathcal{E}^{t}$ and $j\in F^{t},$
otherwise $q_{t}(i,j)=0.$ That is to say that only essential players can reduce
their fan costs in a given period. Alternatively, the amount of the cost $q_{t}(i,j)$
can be interpreted as the maximum cost increase that a fan $j\in F^{t}$ is able to assume, in a certain period $t,$ to incentivize the essential player $i\in \mathcal{E}^{t}.$
\medskip

Next we define a new cost allocation rule, the Omega point, that considers the maximum cost increase mentioned above.

\begin{definition}[Omega point]
Let $(N,D,\Re )\in \Upsilon $ and $(N,c)$ be the corresponding PI-game.
The Omega point $\omega \in \mathbb{R}^{n}$ is defined as
$\omega _{i}=\sum_{t=1}^{T}\omega _{i}^{t}$ for all player $i\in N,$
where for each period $t=1,\ldots ,T$,
\begin{equation*}
\omega _{i}^{t}:=\left\{
\begin{array}{cc}
y_{t}^{\ast }(N)d_{t}^{i}+\overset{Q_{i}^{t}}{\overbrace{q_{t}(\mathcal{E}%
^{t},i)-\sum_{j\in F^{t}}q_{t}(i,j)}} & if\text{ \ }\left\vert \mathcal{E}%
^{t}\right\vert =1 \\
y_{t}^{\ast }(N)d_{t}^{i} & \text{otherwise}%
\end{array}%
\right. \text{ }
\end{equation*}
\end{definition}

The Omega point means that, in each of the periods with a single essential player, i.e. without competition, this essential player gets a cost reduction from his fans. The amount $Q_{i}^{t}$ represents the cost reduction or increase, depending on the sign, for player $i\in N$ in the period $t$. Notice that $Q_{i}^{t}<0,$ only if $i$ is an essential player, otherwise $Q_{i}^{t}\geq 0.$ In addition, $Q^{t}=\left( Q_{i}^{t}\right) _{i\in N}$ for all $t\in
\{1,...,T\}.$

The reader may also note that  $\omega =o+Q,$ where $Q\in \mathbb{R}^{n}$ with $Q_{i}=\sum_{t=1}^{T}Q_{i}^{t}.$  It is worth noting that $Q_{i}$ represents the marginal cost reduction or increase, of player $i$ to the rest players.  Moreover, $\sum_{i\in N}Q_{i}=0.$ In this setting, those players with $Q_{i}<0,$ would prefer the Omega point to the Owen point. On the contrary, those players with $Q_{i}>0$  would like the Owen point more.
\smallskip

The following example illustrates the cost reduction that the Omega point applies to essential players while increasing the cost of fans.

\begin{example}
In example \ref{example 1} the Owen point is $o=(25,26,13)$. Moreover, $%
\mathcal{E}^{1}=\{1\}$, $F^{1}=\{2,3\},\mathcal{E}^{2}=\mathcal{E}%
^{3}=\varnothing $ . The cost reduction for the essential player $1$ from his fans $2$ and $3$ are:

\begin{equation*}
q_{1}(1,2)=8;q_{1}(1,3)=6;
\end{equation*}

Therefore,
\begin{eqnarray*}
\omega _{1} &=&o_{1}-q_{1}(1,2)-q_{1}(1,3)=25-8-6=11 \\
\omega _{2} &=&o_{2}+q_{1}(1,2)=34 \\
\omega _{3} &=&o_{3}+q_{1}(1,3)=19
\end{eqnarray*}

In this case $\omega =(11,34,19)=(25,26,13)+Q$, with $Q=(-14,8,6)$.
It is also a core-allocation. Note that player 1 obtains a cost reduction of 14
units, while players 2 and 3 are increasing their costs by 8 and 6 units, respectively.
Here, the Omega point is a core-allocation that recognizes the essential role of player 1
through a cost reduction assumed by his fans. Next we demonstrate that this always holds
for any PI-game.
\end{example}

\begin{proposition}
\label{wcore}Let $(N,D,\Re )\in \Upsilon $ and $(N,c)$ be the corresponding
PI-game. The Omega point is a core-allocation.
\end{proposition}

\begin{proof}
Consider any period $t$ and a coalition $S\subseteq N$. If $t$ does not have
essential players or has more than one, then $\omega _{S}^{t}=\sum_{i\in
S}y_{t}^{\ast }(N)d_{t}^{i}\leq \sum_{i\in S}y_{t}^{\ast
}(S)d_{t}^{i}=y_{t}^{\ast }(S)d_{t}^{S}$.
\medskip

Otherwise, suppose that player $k$ is essential in the period $t$ ($\mathcal{%
E}^{t}=\{k\}$ ). we distinguish two possibilities:

\begin{itemize}
\item $k$ $\in S$, then%
\begin{eqnarray*}
\omega _{S}^{t} &=&\omega _{k}^{t}+\sum_{i\in S\cap F^{t}}\omega
_{i}^{t}+\sum_{i\in S\backslash F^{t}}\omega _{i}^{t}=y_{t}^{\ast
}(N)d_{t}^{k}-\sum_{j\in F^{t}}q_{t}(k,j) \\
&&+\sum_{i\in S\cap F^{t}}\left( y_{t}^{\ast }(N)d_{t}^{i}+q_{t}(k,i)\right)
+\sum_{i\in S\backslash F^{t}}y_{t}^{\ast }(N)d_{t}^{i} \\
&=&y_{t}^{\ast }(N)d_{t}^{S}+\sum_{i\in S\cap F^{t}}q_{t}(k,i)-\sum_{j\in
F^{t}}q_{t}(k,j) \\
&=&y_{t}^{\ast }(N)d_{t}^{S}-\sum_{j\in F^{t}\backslash S}q_{t}(k,j)\leq
y_{t}^{\ast }(N)d_{t}^{S}\leq y_{t}^{\ast }(S)d_{t}^{S}
\end{eqnarray*}%
\newline

\item $k$ $\notin $ $S$, then%
\begin{eqnarray*}
\omega _{S}^{t} &=&\sum_{i\in S\cap F^{t}}\omega _{i}^{t}+\sum_{i\in
S\backslash F^{t}}\omega _{i}^{t}=\sum_{i\in S\cap F^{t}}\left( y_{t}^{\ast
}(N)d_{t}^{i}+q_{t}(k,i)\right) +\sum_{i\in S\backslash F^{t}}y_{t}^{\ast
}(N)d_{t}^{i} \\
&=&y_{t}^{\ast }(N)d_{t}^{S}+\sum_{i\in S\cap F^{t}}\left( y_{t}^{\ast
}(N\backslash \{k\})-y_{t}^{\ast }(N)\right) d_{t}^{i} \\
&=&\sum_{i\in S\backslash F^{t}}y_{t}^{\ast }(N)d_{t}^{i}+\sum_{i\in S\cap
F^{t}}y_{t}^{\ast }(N\backslash \{k\})d_{t}^{i} \\
&\leq &\sum_{i\in S}y_{t}^{\ast }(N\backslash \{k\})d_{t}^{i}\leq \sum_{i\in
S}y_{t}^{\ast }(S)d_{t}^{i}=y_{t}^{\ast }(S)d_{t}^{S}
\end{eqnarray*}
\end{itemize}

Hence, $\omega _{S}^{t}\leq y_{t}^{\ast }(S)d_{t}^{S}$ for all $t\in T$.
Then, $\omega _{S}=\sum_{t=1}^{T}\sum_{i\in S}\omega
_{i}^{t}=\sum_{t=1}^{T}\omega _{S}^{t}\leq \sum_{t=1}^{T}y_{t}^{\ast
}(S)d_{t}^{S}=c(S)$ for any coalition $S\subseteq N$. Moreover, $\omega
_{N}=o_{N}+\sum_{i\in N}Q_{i}=o_{N}=c(N).$ Therefore $\omega \in Core(N,c)$.
\end{proof}

\subsection{\protect\bigskip Characterization of the Omega point}

To complete the study of the Omega point, we here propose an axiomatic characterization based on a set of desirable properties that make it unique. In order to do that, we denote by $\varphi $ a
generic allocation rule on $\Upsilon $ and consider the following properties,  some of which have been used in the literature to axiomatize alternative allocations:

\begin{itemize}
\item[(EF)] \textit{Efficiency}. For all $x\in \varphi (N,D,\Re )$ and for
any PI-game $(N,D,\Re )\in \Upsilon $, $x_{N}=c^{(N,D,\Re )}(N)$.

\item[(NE)] \textit{Nonemptiness}. For any PI-game $(N,D,\Re )\in \Upsilon $%
, $\varphi (N,D,\Re )\neq \varnothing $.

\item[(IBC)] \textit{Inessential {bounded} cost}. For any PI-game $(N,D,\Re
)\in \Upsilon $ and for all $x\in \varphi (N,D,\Re )$, if $i$ is an
inessential player for $(N,D,\Re )$, then $x_{i}\leq
\sum_{t=1}^{T}y_{t}^{\ast }(N\backslash \mathcal{E}^{t})d_{t}^{i}$

\item[(TI)] \textit{Tyranny}. For all $(N,D,\Re )\in \Upsilon $ and for all $%
x\in \varphi (N,D,\Re )$, if $k$ is a single essential player then $%
x_{N\setminus \{k\}}=c^{(N,D,\Re )}(N\setminus \{k\})$.

\item[(ACP)] \textit{{Additive} combination of periods' demands}. For all $%
(N,D,\Re )\in \Upsilon $ and for all $x\in \varphi (N,D,\Re )$, there exists
$(z_{t})_{t\in T}\in (\mathbb{R}^{N})^{N}$ such that $x=\sum_{t=1}^{T}z_{t}$
and for all $t\in T$, $z_{t}\in \varphi (N,D_{t},\Re )$ if $\left\vert
\mathcal{E}^{t}\right\vert \leq 1$ and $z_{t}=Owen(N,D_{t},\Re )$ otherwise,
where
\begin{equation}
D_{t}=\left( d^{ip}\right) _{\substack{ i=1,...,n  \\ p=1,...,T}},\text{ }%
d^{ip}=\left\{
\begin{array}{cc}
d_{p}^{i} & \text{if }t=p, \\
0 & \text{otherwise.}%
\end{array}%
\right.  \label{eq3}
\end{equation}
\end{itemize}

The first two properties were already used in Guardiola et al. (2008), among
many other papers, to characterize the Owen point solution, and they are also
important to our new characterization of the Omega point. Recall that
\textit{Efficiency} ensures that the total cost of any PI-situation is
entirely allocated among the players. Analogously, \textit{Nonemptiness}
guarantees that this allocation rule always return a feasible allocation of
the overall cost when applied to any PI-situation. \textit{Inessential {%
bounded} cost} imposes a maxim cost for every {inessential} player in
situations which an essentials players has left. \textit{Tyranny} implies
that a single essential player will assert all his power over the rest so
that they assume the maximum possible cost.
\smallskip

Finally, an allocation rule satisfies the property of \textit{Additivity
combination of periods' demands} if it is additive with respect to the
demand of the periods that has at most an essential player plus the Owen
point of those periods with more than one essential player.
We emphasize that this additivity results from the following
relationship $c^{(N,D,\Re )}=\sum_{t=1}^{T}c^{(N,D_{t},\Re )}$\textbf{\ }for
all\textbf{\ }$(N,D,\Re )\in \Upsilon .$ Thus, we are interested on allocation
rules, for PI-situations, compatible with this form of distribution of their
demands.
\medskip

First, we prove that the Omega point satisfies all the properties mentioned above.

\begin{proposition}
\label{con} The Omega point defined on the set $\Upsilon ,$ satisfies EF, NE, IBC, TI and ACP.
\end{proposition}

\begin{proof}
For any PI situation $(N,D,\Re )\in \Upsilon $ we know by proposition \ref%
{wcore} that $\omega (N,D,\Re )\in Core(N,c)$ by . Hence, the Omega point
verifies the properties of EF and NE. An inessential player satisfy IBC
since if for all $i\in N$
\begin{eqnarray*}
\omega _{i}(N,D,\Re ) &=&\sum_{t=1}^{T}\omega _{i}^{t}(N,D,\Re
)=\sum_{t=1}^{T}\omega _{i}(N,D_{t},\Re ) \\
&=&\sum_{t\in T/\left\vert \mathcal{E}^{t}\right\vert =1}\left( y_{t}^{\ast
}(N)d_{t}^{i}+q_{t}(\mathcal{E}^{t},i)\right) +\sum_{t\in T/\left\vert
\mathcal{E}^{t}\right\vert \neq 1}y_{t}^{\ast }(N)d_{t}^{i} \\
&=&\sum_{t\in T/\left\vert \mathcal{E}^{t}\right\vert =1}y_{t}^{\ast
}(N\backslash \mathcal{E}^{t})d_{t}^{i}+\sum_{t\in T/\left\vert \mathcal{E}%
^{t}\right\vert \neq 1}y_{t}^{\ast }(N)d_{t}^{i}\leq
\sum_{t=1}^{T}y_{t}^{\ast }(N\backslash \mathcal{E}^{t})d_{t}^{i}
\end{eqnarray*}

\noindent if there is only one essential player $k$, then:
\begin{eqnarray*}
\sum_{t=1}^{T}\omega _{N\setminus \{k\}}^{t} &=&\sum_{t=1}^{T}\left(
y_{t}^{\ast }(N)d_{t}^{N\setminus \{k\}}+\sum_{j\in N\setminus
\{k\}}q_{t}(k,i)\right) \\
&=&\sum_{t=1}^{T}\left( y_{t}^{\ast }(N)d_{t}^{N\setminus \{k\}}+\sum_{j\in
N\setminus \{k\}}\left( y_{t}^{\ast }(N\backslash \{k\})-y_{t}^{\ast
}(N)\right) d_{t}^{j}\right) \\
&=&\sum_{t=1}^{T}\left( y_{t}^{\ast }(N\backslash \{k\})d_{t}^{N\setminus
\{k\}}\right) =c^{(N,D,\Re )}(N\setminus \{k\}).
\end{eqnarray*}%
Then satisfy {TI}. Finally, considering $D_{t}$ as it was already defined in
(\ref{eq3}), we obtain that $\sum_{t=1}^{T}D_{t}=D$ and
\begin{eqnarray*}
\omega (N,D,\Re ) &=&\left( \sum_{t=1}^{T}\omega _{i}^{t}(N,D,\Re )\right)
_{i\in N}=\sum_{t=1}^{T}\left( \omega _{i}(N,D_{t},\Re )\right) _{i\in N} \\
&=&\sum_{t\in T/\left\vert \mathcal{E}^{t}\right\vert \leq 1}\omega
_{i}(N,D_{t},\Re )+\sum_{t\in T/\left\vert \mathcal{E}^{t}\right\vert \geq
2}Owen(N,D_{t},\Re )
\end{eqnarray*}%
Hence, the Omega point satisfies ACP.
\end{proof}

\medskip
Second, we focus on PI-situations without essential players and show that, in this setting,
the Omega point matches the Owen point, and both can be characterized by using only three of the previous properties.

\begin{proposition}
\label{con2} Let $(N,D,\Re )\in \Upsilon $ be a PI situation with $\left\vert \mathcal{E}\right\vert
=0$. Then, $\varphi (N,D,\Re )=\omega (N,D,\Re )=Owen(N,D,\Re )$ if and only
if $\varphi $ satisfies NE, EF and IBC.
\end{proposition}

\begin{proof}
(If) Immediately follows by Proposition \ref{con}. \medskip

\noindent (Only if) By NE, $\varphi (N,D,\Re )\neq \varnothing $. Take $x\in
\varphi (N,D,\Re ).$ Since all players $i\in N$ are inessential, by IBC, it
holds that $x_{i}\leq \sum_{t=1}^{T}y_{t}^{\ast }(N\backslash \mathcal{E}%
^{t})d_{t}^{i}$ $=\sum_{t=1}^{T}y_{t}^{\ast }(N)d_{t}^{i}=\omega
_{i}(N,D,\Re )$ for each $i\in N$ $.$ Therefor,e by EF, $\varphi (N,D,\Re
)=\omega (N,D,\Re )=Owen(N,D,\Re )$.\medskip
\end{proof}

The main Theorem of this section shows that the Omega point is
the unique allocation rule that satisfies the aforementioned five properties.

\begin{theorem}
\label{charac} An allocation rule on $(N,D,\Re )\in \Upsilon$\ satisfies the
properties EF, NE, IBC, TI and ACP if and only if it coincides with the
Omega point.
\end{theorem}

\begin{proof}
(If) The \textit{if} part of the proof is direct from Proposition \ref{con}.
\medskip

\noindent (Only if) Let $\varphi $ be an allocation rule. The case where the
number of essential players is zero, namely $\left\vert \mathcal{E}%
\right\vert =0$, follows from Proposition \ref{con2}. Then, it remains to
prove the case when $\left\vert \mathcal{E}\right\vert \geq 1.$In this case,
we know that $D=D_{1}+D_{2}+...+D_{T}$ where $D_{t}$ is (see (\ref{eq3}%
)):
\begin{equation*}
D_{t}=\left(
\begin{array}{ccccccc}
0 & \dots & 0 & d_{t}^{1} & 0 & \dots & 0 \\
0 & \dots & 0 & d_{t}^{2} & 0 & \dots & 0 \\
\vdots & \dots & 0 & \vdots & 0 & \dots & \vdots \\
0 & \dots & 0 & d_{t}^{n} & 0 & \dots & 0%
\end{array}%
\right).
\end{equation*}%
Then for all $t\in T$, $(N,D_{t},\Re )$ is a PI-situation with $D_{t}$ an
integer matrix. This implies that $(N,D_{t},\Re )$ belongs to $\Upsilon $.
Therefore, for any $t\in T,$ the Omega point for $(N,D_{t},\Re )$ is $%
(\omega _{i}(N,D_{t},\Re ))_{i=1,\ldots ,n}$:

By NE, $\varphi (N,D_{t},\Re )\neq \varnothing $. for each situation $%
(N,D_{t},\Re )$ we have two cases:

\begin{itemize}
\item $\left\vert \mathcal{E}^{t}\right\vert =0,$ then by Proposition \ref%
{con2} $\varphi (N,D_{t},\Re )=\omega (N,D_{t},\Re )=Owen(N,D_{t},\Re ).\ $

\item $\mathcal{E}^{t}=\{k\}.$ Take $u\in \varphi (N,D_{t},\Re )$, by IBC $%
u_{i}\leq y_{t}^{\ast }(N\backslash \{k\})d_{t}^{i}$ for all $i\in
N\backslash \{k\}$ and by TY $u_{N\backslash \{k\}}=y_{t}^{\ast
}(N\backslash \{k\})d_{t}^{N\setminus \{k\}}$. Hence for all $i\in
N\backslash \{k\}$ $u_{i}=y_{t}^{\ast }(N\backslash
\{k\})d_{t}^{i}=y_{t}^{\ast }(N)d_{t}^{i}+q_{t}(\mathcal{E}^{t},i)=\omega
_{i}(N,D_{t},\Re )$. Finally, by EF $u_{k}=c(N)-c(N\backslash
\{k\})=y_{t}^{\ast }(N)d_{t}^{N}-y_{t}^{\ast }(N\backslash
\{k\})d_{t}^{N\setminus \{k\}}=y_{t}^{\ast }(N)d_{t}^{k}-\left( y_{t}^{\ast
}(N\backslash \{k\})-y_{t}^{\ast }(N)\right) d_{t}^{N\backslash
\{k\}}=y_{t}^{\ast }(N)d_{t}^{k}-\sum_{j\in F^{t}}q_{t}(k,j)=\omega
_{k}(N,D_{t},\Re )$.
\end{itemize}

Therefore, if $x\in \varphi (N,D,\Re )$ by ACP one has that $%
x=z_{1}+...+z_{t}$ with $z_{t}\in \varphi (N,D_{t},\Re )$ for all $t\in T$,
and so
\begin{equation*}
x=\sum_{t=1}^{T}z_{t}=\sum_{t\in T/\left\vert \mathcal{E}^{t}\right\vert
\leq 1}\omega (N,D_{t},\Re )+\sum_{t\in T/\left\vert \mathcal{E}%
^{t}\right\vert \geq 2}Owen(N,D_{t},\Re )=\omega (N,D,\Re ).
\end{equation*}%
The above equation implies that $\varphi (N,D,\Re )=\omega (N,D,\Re )$.
\end{proof}
\bigskip

Finally, we prove that all the properties used in Theorem \ref{charac} are logically independent.
That is, the characterization of the Omega point is tight in the sense that no property is redundant.

\begin{example}
Let $\varphi $ be a solution rule defined on $\Upsilon $ as
\begin{equation*}
\varphi (N,D,\Re ):=\left\{
\begin{array}{cc}
\left( \frac{c^{(N,D,\Re )}(N)}{2},\frac{c^{(N,D,\Re )}(N)}{2}\right) , &
(N,D,\Re )\in \Upsilon ^{1} \\
&  \\
\omega (N,D,\Re ), & \text{otherwise,}%
\end{array}%
\right.
\end{equation*}%
\noindent where
\begin{equation*}
\Upsilon ^{1}:=\left\{ (N,D,\Re )\in \Upsilon \left/ \left\vert N\right\vert
=2,T=2,\mathcal{E}^{1}=\{1,2\},\mathcal{E}^{2}\text{ }=\varnothing \right.
\right\} .
\end{equation*}%
\noindent $\varphi (N,D,\Re )$ satisfies EF, NE, IBC and TI, but not ACP.
\end{example}

\begin{example}
Let $\varphi $ be a solution rule defined on $\Upsilon $ as
\begin{equation*}
\varphi (N,D,\Re ):=\left\{
\begin{array}{cc}
\left( c^{(N,D,\Re )}(N),0\right) , & (N,D,\Re )\in \Upsilon ^{2} \\
&  \\
\omega (N,D,\Re ), & \text{otherwise,}%
\end{array}%
\right.
\end{equation*}%
\noindent where
\begin{equation*}
\Upsilon ^{2}:=\left\{ (N,D,\Re )\in \Upsilon \left/ \left\vert N\right\vert
=2,T=1,\mathcal{E}=\varnothing \text{ }\right. \right\} .
\end{equation*}

\noindent {$\varphi $}$(N,D,\Re )${\ } satisfies EF, NE, ACP and TI, but not
IBC.
\end{example}

\begin{example}
Let $\varphi $ be a solution rule defined on $\Upsilon $ as

\begin{equation*}
\varphi (N,D,\Re ):=Owen(N,D,\Re )
\end{equation*}

\noindent {$\varphi $}$(N,D,\Re )${\ } satisfies EF, NE, IBC, and ACP, but
not TI.
\end{example}

\begin{example}
Let $\varphi $ be a solution rule defined on $\Upsilon $ as
\begin{equation*}
\varphi (N,D,\Re ):=\left\{
\begin{array}{cc}
\left( c^{(N,D,\Re )}(N\setminus \{\mathcal{E}^{1}\}),c^{(N,D,\Re
)}(N\setminus \{\mathcal{E}^{1}\})\right) , & (N,D,\Re )\in \Upsilon ^{3} \\
&  \\
\omega (N,D,\Re ), & \text{otherwise,}%
\end{array}%
\right. ,
\end{equation*}%
where%
\begin{equation*}
\Upsilon ^{3}:=\left\{ (N,D,\Re )\in \Upsilon \left/ \left\vert N\right\vert
=2,T=1,\left\vert \mathcal{E}^{1}\right\vert =1\text{ }\right. \right\} .
\end{equation*}

\noindent {$\varphi $}$(N,D,\Re )${\ } satisfies NE, IBC, TI and ACP, but
not EF.
\end{example}

\begin{example}
Let $\varphi $ be a solution rule defined on $\Upsilon $ as
\begin{equation*}
\varphi (N,D,\Re ):=\varnothing .
\end{equation*}

\noindent {$\varphi $}$(N,D,\Re )${\ satisfies }EF, ACP, IBC, and TI. but
not NE.
\end{example}

\section{\protect\bigskip \label{QPQ set}Quid Pro Quo allocations}

As we already mentioned, the Omega point can be considered the natural aspiration of
the essential players to achieve the biggest cost reduction while the Owen
point reflects their altruistic character. We combine both extreme characteristics and
define the $\lambda$-agreement $a(\lambda ):=\lambda \omega +\left( 1-\lambda \right)
o $ with $\lambda \in \left[ 0,1\right] ,$ as the convex linear combination
of the Owen point and the Omega point. The parameter $\lambda $ represents
here the weight given to individual
behavior, by those players who want to maximize their cost reduction,
compared to altruistic behavior (by $1-\lambda$), which benefits the other players.
\smallskip

The set of all the above agreements is called Quid Pro Quo allocation set.

\begin{definition}[Quid Pro Quo allocation set]
\bigskip Let $(N,D,\Re )\in \Upsilon $ and $(N,c)$ be the corresponding
PI-game. We define the Quid pro quo allocation set as follows:
\begin{equation*}
QPQ(N,c):=\left\{ a(\lambda )\ \text{such that }\lambda \in \left[ 0,1\right]
\right\} .
\end{equation*}
\end{definition}

The Quid Pro Quo allocation set, henceforth QPQ-set, is a parametric family of core-allocations.
That is, $QPQ(N,c)\subseteq Core(N,c)$.
\medskip

The following example illustrate the wealth of the QPQ set of a PI-situation with multiple  essential
players.

\begin{example}
\label{example 3} Let us consider a PI-situation with four players in four
periods with demand, and production, inventory and backlogging costs given
in the following table:

\begin{equation*}
\begin{tabular}{|c|c|c|c|c||c|c|c|c||c|c|c|c||c|c|c|c|}
\hline
& \multicolumn{4}{|c}{Demand} & \multicolumn{4}{||c||}{Production} &
\multicolumn{4}{|c||}{Inventory} & \multicolumn{4}{||c|}{Backlogging} \\
\hline
P1 & 2 & 1 & 2 & 2 & 1 & 2 & 2 & 2 & 1 & 1 & 1 &  & 2 & 2 & 2 &  \\ \hline
P2 & 2 & 2 & 1 & 2 & 2 & 1 & 2 & 2 & 1 & 1 & 1 &  & 2 & 2 & 2 &  \\ \hline
P3 & 2 & 1 & 2 & 2 & 2 & 2 & 1 & 2 & 1 & 1 & 1 &  & 2 & 2 & 2 &  \\ \hline
P4 & 2 & 1 & 1 & 2 & 2 & 2 & 2 & 1 & 1 & 1 & 1 &  & 2 & 2 & 2 &  \\ \hline
\end{tabular}%
\end{equation*}%
\newline
\newline
The above table described a cooperative game with a characteristic function
detailed in the following table:

\begin{equation*}
\begin{array}{|c|c|c|c|c||c|c|c|c||c|c|c||c|c|c||c||}
\hline
& d_{1}^{_{S}} & d_{2}^{_{S}} & d_{3}^{_{S}} & d_{4}^{_{S}} & p_{1}^{_{S}} &
p_{2}^{_{S}} & p_{3}^{_{S}} & p_{4}^{_{S}} & h_{1}^{_{S}} & h_{2}^{_{S}} &
h_{3}^{_{S}} & b_{1}^{_{S}} & b_{2}^{_{S}} & b_{3}^{_{S}} & c \\ \hline
\{1\} & 2 & 1 & 2 & 2 & 1 & 2 & 2 & 2 & 1 & 1 & 1 & 2 & 2 & 2 & 12 \\ \hline
\{2\} & 2 & 2 & 1 & 2 & 2 & 1 & 2 & 2 & 1 & 1 & 1 & 2 & 2 & 2 & 12 \\ \hline
\{3\} & 2 & 1 & 2 & 2 & 2 & 2 & 1 & 2 & 1 & 1 & 1 & 2 & 2 & 2 & 12 \\ \hline
\{4\} & 2 & 1 & 1 & 2 & 2 & 2 & 2 & 1 & 1 & 1 & 1 & 2 & 2 & 2 & 10 \\ \hline
\{1,2\} & 4 & 3 & 3 & 4 & 1 & 1 & 2 & 2 & 1 & 1 & 1 & 2 & 2 & 2 & 21 \\
\hline
\{1,3\} & 4 & 2 & 4 & 4 & 1 & 2 & 1 & 2 & 1 & 1 & 1 & 2 & 2 & 2 & 20 \\
\hline
\{1,4\} & 4 & 2 & 3 & 4 & 1 & 2 & 2 & 1 & 1 & 1 & 1 & 2 & 2 & 2 & 18 \\
\hline
\{2,3\} & 4 & 3 & 3 & 4 & 2 & 1 & 1 & 2 & 1 & 1 & 1 & 2 & 2 & 2 & 22 \\
\hline
\{2,4\} & 4 & 3 & 2 & 4 & 2 & 1 & 2 & 1 & 1 & 1 & 1 & 2 & 2 & 2 & 19 \\
\hline
\{3,4\} & 4 & 2 & 3 & 4 & 2 & 2 & 1 & 1 & 1 & 1 & 1 & 2 & 2 & 2 & 19 \\
\hline
\{1,2,3\} & 6 & 4 & 5 & 6 & 1 & 1 & 1 & 2 & 1 & 1 & 1 & 2 & 2 & 2 & 27 \\
\hline
\{1,2,4\} & 6 & 4 & 5 & 6 & 1 & 1 & 2 & 1 & 1 & 1 & 1 & 2 & 2 & 2 & 24 \\
\hline
\{1,3,4\} & 6 & 3 & 5 & 6 & 1 & 2 & 1 & 1 & 1 & 1 & 1 & 2 & 2 & 2 & 23 \\
\hline
\{2,3,4\} & 6 & 4 & 4 & 6 & 2 & 1 & 1 & 1 & 1 & 1 & 1 & 2 & 2 & 2 & 26 \\
\hline
\{1,2,3,4\} & 8 & 5 & 6 & 8 & 1 & 1 & 1 & 1 & 1 & 1 & 1 & 2 & 2 & 2 & 27 \\
\hline
\end{array}%
\end{equation*}

Here, $y^{\ast }(N)=\left( 1,1,1,1\right) $ and the Owen point
is $o=(7,7,7,6)$. Moreover, $\mathcal{E}=N,$ because of each player is essential
just in one period. For example in period 1, $\mathcal{E}^{1}=\{1\}$ and $%
F^{1}=\{2,3,4\}$. In addition, $q_{1}(1,j)=2$ for $j\in F^{1}$, $%
Q^{1}=(-6,2,2,2),Q^{2}=(1,-3,1,1),Q^{3}=(2,1,-4,1)$ and $Q^{4}=(2,2,2,-6).$
\medskip

It is easy to check that $Q=(-1,2,1,-2)$ that is, players 1 and 4 are interested in improving
the Owen point, and they would prefer the Omega point. However, players 2 and 3, still being essential, get some benefit with the Owen point's and they would prefer to keep on it.
\smallskip

On the other hand, here the omega point is $\omega =(6,9,8,4)$ and the QPQ set is given by:
\begin{equation*}
QPQ(N,c):=\left\{ (7-\lambda ,7+2\lambda ,7+\lambda ,6-2\lambda )\ \text{%
such that }\lambda \in \left[ 0,1\right] \right\}
\end{equation*}

If we consider the same weight for both individual and altruistic behaviors, we get the Shapley,
which also matches the Nucleolus. That is, for $\lambda =\frac{1}{2}$ the Shapley
value and Nucleolus coincides and both are equal to $\left( \frac{13}{2},8,%
\frac{15}{2},5\right) $.
\end{example}
\medskip

At this point we wonder whether this coincidence always holds for every PI-game. The
answer is no, in general, as example \ref{example 4} reveals.
\medskip

The main result of this section shows that, if no player can get a cost reduction
in any coalition without an essential player, then the equal agreement, $a\left( \frac{1}{2}\right) $,
coincides with the Shapley value and the Nucleolus. In some sense, it is a Solomonic
agreement between the players who demand cost reductions (individual behaviour) and those who do
not (altruistic behaviour). For that, we call $a\left( \frac{1}{2}\right) $ Solomonic allocation
and denote it $\varsigma (N,c).$

\begin{proposition}
Let $(N,D,\Re )\in \Upsilon $ and $(N,c)$ be the corresponding PI-game. Assume that
for each $t=1,\ldots ,T$ the following conditions are simultaneously fulfilled:

\begin{enumerate}
\item[(i)] $\left\vert \mathcal{E}^{t}\right\vert \leq 1$,

\item[(ii)] $y_{t}^{\ast }(\mathcal{E}^{t})=y_{t}^{\ast }(N)$ \ if $\mathcal{%
E}^{t}\neq \phi ,$

\item[(iii)] $y_{t}^{\ast }(N\backslash \mathcal{E}^{t})=y_{t}^{\ast
}(\{i\}) $ for all $i\in N\backslash \mathcal{E}^{t}.$ \medskip

\noindent Then, $\varsigma (N,c)=$ $\phi (N,c)=\eta (N,c).$
\end{enumerate}

\begin{proof}
Consider $(N,D_{t},\Re )\in \Upsilon $ and $(N,c^{t})$ be the corresponding
PI-game, with, (only period $t$ has demand)%
\begin{equation*}
D_{t}=\left( d^{ip}\right) _{\substack{ i=1,...,n  \\ p=1,...,T}},\text{ }%
d^{ip}=\left\{
\begin{array}{cc}
d_{p}^{i} & \text{if }t=p, \\
0 & \text{otherwise.}%
\end{array}%
\right.
\end{equation*}

We will denote to simplify notation $o(N,D_{t},\Re )$ and $\omega
(N,D_{t},\Re )$ as $o^{t}$ and $\omega ^{t}$, respectively. By (i) we
consider only two cases:

\begin{itemize}
\item If $\left\vert \mathcal{E}^{t}\right\vert =0$ then $\omega
^{t}=o^{t}=Core(N,c^{t})=\eta (N,c^{t})$ since the Nucleolus always belongs
to the core of a game. Moreover, because of the condition (iii) $y_{t}^{\ast
}(N)=y_{t}^{\ast }(\{i\})$ for all $i\in N,$ then $c^{t}(S)=o_{S}^{t}$ for
all $S\subseteq N,$ It is easy to verify that all players are dummy players
then $\phi (N,c^{t})=o^{t}.$

\item If $\left\vert \mathcal{E}^{t}\right\vert =1,\ (\mathcal{E}%
^{t}=\{k\}). $

Note that if $k\in S$ then $c^{t}(S\cup \{i\})-c^{t}(S)=y_{t}^{\ast
}(N)d_{t}^{i},$ otherwise $(k\notin S)$ then by condition (iii) $c^{t}(S\cup
\{i\})-c^{t}(S)=y_{t}^{\ast }(S\cup \{i\})d_{t}^{S\cup \{i\}}-y_{t}^{\ast
}(S)d_{t}^{S}=y_{t}^{\ast }(N\backslash \mathcal{E}^{t})d_{t}^{i}$ for all $%
i\in N\backslash \mathcal{E}^{t}.$

If $i\in N\backslash \mathcal{E}^{t}$ then,
\begin{eqnarray*}
\phi _{i}(N,c^{t}) &=&\sum_{S\subseteq N\diagdown \{i\}}\gamma (S)\cdot
\left[ c^{t}(S\cup \{i\})-c^{t}(S)\right] = \\
&&\sum_{S\subseteq N\diagdown \{i\}/k\in S}\gamma (S)\cdot \left[
c^{t}(S\cup \{i\})-c^{t}(S)\right] \\
&&+\sum_{S\subseteq N\diagdown \{i\}/k\notin S}\gamma (S)\cdot \left[
c^{t}(S\cup \{i\})-c^{t}(S)\right] \\
&=&\sum_{S\subseteq N\diagdown \{i\}/k\in S}\gamma (S)\cdot y_{t}^{\ast
}(N)d_{t}^{i}+\sum_{S\subseteq N\diagdown \{i\}/k\notin S}\gamma (S)\cdot
y_{t}^{\ast }(N\backslash \mathcal{E}^{t})d_{t}^{i} \\
&=&y_{t}^{\ast }(N)d_{t}^{i}\cdot \left( \sum_{S\subseteq N\diagdown
\{i\}/k\in S}\gamma (S)\right) \\
&&+y_{t}^{\ast }(N\backslash \mathcal{E}^{t})d_{t}^{i}\cdot \left(
\sum_{S\subseteq N\diagdown \{i\}/k\notin S}\gamma (S)\right) \text{ } \\
&=&\frac{1}{2}\cdot y_{t}^{\ast }(N)d_{t}^{i}+\frac{1}{2}\cdot y_{t}^{\ast
}(N\backslash \mathcal{E}^{t})d_{t}^{i} \\
&=&\frac{1}{2}\cdot y_{t}^{\ast }(N)d_{t}^{i}+\frac{1}{2}\cdot \left(
y_{t}^{\ast }(N)d_{t}^{i}+\left( y_{t}^{\ast }(N\backslash \{\mathcal{E}%
^{t}\})-y_{t}^{\ast }(N)\right) \cdot d_{t}^{i}\right) \\
&=&\frac{1}{2}\cdot o_{i}^{t}+\frac{1}{2}\cdot \omega _{i}^{t}
\end{eqnarray*}

By efficiency of Shapley value $\phi _{\mathcal{E}^{t}}(N,c^{t})=\frac{1}{2}%
\cdot o_{\mathcal{E}^{t}}^{t}+\frac{1}{2}\cdot \omega _{\mathcal{E}%
^{t}}^{t}. $ Moreover, Shapley value {satisfies additivity} property, thus
for all player $i\in N$
\begin{eqnarray*}
\phi _{i}(N,c) &=&\sum_{t=1}^{T}\phi _{i}(N,c^{t})=\sum_{t=1}^{T}\left(
\frac{1}{2}\cdot o_{i}^{t}+\frac{1}{2}\cdot \omega _{i}^{t}\right) \\
&=&\frac{1}{2}\cdot \sum_{t=1}^{T}o_{i}^{t}+\frac{1}{2}\cdot
\sum_{t=1}^{T}\omega _{i}^{t} \\
&=&\frac{1}{2}\cdot o_{i}(N,D,\Re )+\frac{1}{2}\cdot \omega _{i}(N,D,\Re )
\end{eqnarray*}

since the Owen point is additive for the demands (demonstrated in Guardiola
et al. (2008)) and $\omega _{i}(N,D_{t},\Re )=\omega _{i}^{t}(N,D,\Re )$.
Hence, $\varsigma (N,c)=\phi (N,c).$

Now, we will prove that the {Shapley} value coincides with the Nucleolus. As
we have seen previously if the properties $(i),(ii)$ and $(iii)$ are
satisfied for a period $t=1,\ldots ,T$ and for each $i\in N$ and for all $%
S\subseteq N\setminus \{i\}$
\begin{equation*}
\Delta _{i}^{t}(S):=c^{t}(S\cup \{i\})-c^{t}(S)=\left\{
\begin{array}{ccc}
y_{t}^{\ast }(N)d_{t}^{i} & if & \mathcal{E}^{t}\in S\text{ and }i\notin
\mathcal{E}^{t} \\
y_{t}^{\ast }(N\backslash \mathcal{E}^{t})d_{t}^{i} & if & \mathcal{E}%
^{t}\notin S\text{ and }i\notin \mathcal{E}^{t} \\
y_{t}^{\ast }(N)d_{t}^{i}-y_{t}^{\ast }(N\backslash \mathcal{E}^{t})d_{t}^{i}
& if & i\in \mathcal{E}^{t}%
\end{array}%
\right.
\end{equation*}

similarly we get that

\begin{equation*}
\Delta _{i}^{t}(N\setminus \left( S\cup \{i\}\right) )=\left\{
\begin{array}{ccc}
y_{t}^{\ast }(N\backslash \mathcal{E}^{t})d_{t}^{i} & if & \mathcal{E}%
^{t}\in S\text{ and }i\notin \mathcal{E}^{t} \\
y_{t}^{\ast }(N)d_{t}^{i} & if & \mathcal{E}^{t}\notin S\text{ and }i\notin
\mathcal{E}^{t} \\
y_{t}^{\ast }(N)d_{t}^{i}-y_{t}^{\ast }(N\backslash \mathcal{E}^{t})d_{t}^{i}
& if & i\in \mathcal{E}^{t}%
\end{array}%
\right.
\end{equation*}

Hence, $\Delta _{i}^{t}(S)+\Delta _{i}^{t}(N\setminus \left( S\cup
\{i\}\right) )=y_{t}^{\ast }(N)d_{t}^{i}+y_{t}^{\ast }(N\backslash \mathcal{E%
}^{t})d_{t}^{i}$ if $i\in N\backslash \mathcal{E}^{t}$ for all $S\subseteq
N\setminus \{i\}$ and $\Delta _{i}^{t}(S)+\Delta _{i}^{t}(N\setminus \left(
S\cup \{i\}\right) )=2\cdot \left( y_{t}^{\ast }(N)d_{t}^{i}-y_{t}^{\ast
}(N\backslash \mathcal{E}^{t})d_{t}^{i}\right) $ if $i\in \mathcal{E}^{t}$
for all $S\subseteq N\setminus \mathcal{E}^{t}.$ We consider $\Delta
_{i}(S):=\sum_{t=1}^{T}\Delta _{i}^{t}(S)$ for each $i\in N$ and for all $%
S\subseteq N\setminus \{i\}.$ Thus $\Delta _{i}(S)+\Delta _{i}(N\setminus
\left( S\cup \{i\}\right) )$ is a constant for all $S\subseteq N\setminus
\{i\}$ and for all $i\in N$. Then $(N,c)$ is a PS-game and $\varsigma (N,c)=$
$\phi (N,c)=\eta (N,c).$
\end{itemize}
\end{proof}
\end{proposition}

The reader may notice that for those situations in which the properties
$(i)$, $(ii)$ {and} $(iii)$ hold and, in addition, $Q=0$ (i.e., $o=\omega$),
then $QPQ(N,c)=\left\{ o(N,D,\Re )\right\} =$ $\left\{ \phi (N,c)\right\}
=\left\{ \eta (N,c)\right\} $. Otherwise, the core is larger.
\medskip

Finally, we analyze the relationships between conditions $(i)$, $(ii)$, $(iii)$
and concavity of PI-games.

\begin{proposition}
Let $(N,D,\Re )\in \Upsilon $ and $(N,c)$ be the corresponding PI-game. If
for each $t=1,\ldots ,T$ conditions (i), (ii) and (ii) are fulfilled
simultaneously the $(N,c)$ is concave.

\begin{proof}
Consider $(N,D_{t},\Re )\in \Upsilon $ and $(N,c^{t})$ be the corresponding
PI-game,

\begin{itemize}
\item[(a)] If $\left\vert \mathcal{E}^{t}\right\vert =0$ then $y_{t}^{\ast
}(\{i\})=y_{t}^{\ast }(N)$ for all $i\in N$ henceforth $c^{t}(S)-c^{t}(S%
\setminus \{i\})=y_{t}^{\ast }(N)d_{t}^{i}$ for all $i\in N$ and for all $%
S\subseteq N.$ Hence $(N,c^{t})$ is concave.
\end{itemize}

\begin{itemize}
\item[(b)] If $\left\vert \mathcal{E}^{t}\right\vert =1,\ $let say $\mathcal{%
E}^{t}=\{k\}$. Then two cases can be {distinguished}:

\item[(b1)] $k\in S$ $\subseteq T\subset N$ then $c^{t}(S)-c^{t}(S\setminus
\{i\})=y_{t}^{\ast }(N)d_{t}^{i}=c^{t}(T)-c^{t}(T\setminus \{i\})$ for all $%
i\in N\setminus \{k\}.$ Finally
\begin{eqnarray*}
c^{t}(S)-c^{t}(S\setminus \{k\}) &\geq &c^{t}(T)-c^{t}(T\setminus \{k\}); \\
y_{t}^{\ast }(N)d_{t}^{S}-y_{t}^{\ast }(N\setminus \{k\})d_{t}^{S\setminus
\{k\}} &\geq &y_{t}^{\ast }(N)d_{t}^{T}-y_{t}^{\ast }(N\setminus
\{k\})d_{t}^{T\setminus \{k\}}; \\
y_{t}^{\ast }(N\setminus \{k\})d_{t}^{T\setminus S} &\geq &y_{t}^{\ast
}(N)d_{t}^{T\setminus S}.
\end{eqnarray*}%
It is true since $y_{t}^{\ast }(S)\geq y_{t}^{\ast }(R)$ for all $S\subseteq
R\subseteq N$\ and all $t\in \{1,...,T\}.$

\item[(b2)] $k\notin S$ and $k\in T.$ By condition (iii) $%
c^{t}(S)-c^{t}(S\setminus \{i\})=y_{t}^{\ast }(N\setminus
\{k\})d_{t}^{i}\geq c^{t}(T)-c^{t}(T\setminus \{i\})$ since if $k\in T$ is {%
satisfied} $c^{t}(T)-c^{t}(T\setminus \{i\})=y_{t}^{\ast }(N)d_{t}^{i}$ and
if $k\notin T$ we have that $c^{t}(T)-c^{t}(T\setminus \{i\})=y_{t}^{\ast
}(N\setminus \{k\})d_{t}^{i}.$
\end{itemize}

Finally, by additivity property of PI-games with respect to periods (see
Guardiola et al. (2008)) $(N,c)$ is concave.
\end{proof}
\end{proposition}

Next example shows that conditions $(i)$, $(ii)$, $(iii)$, although
necessaries, are no sufficient for concavity.

\begin{example}
\label{example 4} Let us consider a PI-situation with three players in three
periods with demand, and production, inventory and backlogging costs given
in the following table:

\begin{equation*}
\begin{tabular}{|c|c|c|c||c|c|c||c|c||c|c|}
\hline
& \multicolumn{3}{|c||}{Demand} & \multicolumn{3}{|c||}{Production} &
\multicolumn{2}{|c||}{Inventory} & \multicolumn{2}{|c|}{Backlogging} \\
\hline
P1 & 10 & 10 & 10 & 1 & 2 & 3 & 1 & 2 & 1 & 1 \\ \hline
P2 & 10 & 10 & 10 & 2 & 1 & 3 & 1 & 2 & 1 & 1 \\ \hline
P3 & 10 & 10 & 10 & 3 & 3 & 1 & 1 & 2 & 2 & 2 \\ \hline
\end{tabular}%
\end{equation*}

Using those data one can obtain the cooperative
game with characteristic function described below:

\begin{equation*}
\begin{array}{|c|c|c|c||c|c|c||c|c||c|c||c||}
\hline
& d_{1}^{_{S}} & d_{2}^{_{S}} & d_{3}^{_{S}} & p_{1}^{_{S}} & p_{2}^{_{S}} &
p_{3}^{_{S}} & h_{1}^{_{S}} & h_{2}^{_{S}} & b_{1}^{_{S}} & b_{2}^{_{S}} & c
\\ \hline
\{1\} & 10 & 10 & 5 & 1 & 2 & 3 & 1 & 2 & 1 & 1 & 45 \\ \hline
\{2\} & 10 & 10 & 10 & 2 & 1 & 3 & 1 & 1 & 1 & 1 & 50 \\ \hline
\{3\} & 10 & 10 & 10 & 3 & 3 & 1 & 1 & 2 & 2 & 2 & 70 \\ \hline
\{1,2\} & 20 & 20 & 15 & 1 & 1 & 3 & 1 & 1 & 1 & 1 & 70 \\ \hline
\{1,3\} & 20 & 20 & 15 & 1 & 2 & 1 & 1 & 2 & 1 & 1 & 75 \\ \hline
\{2,3\} & 20 & 20 & 20 & 2 & 1 & 1 & 1 & 1 & 1 & 1 & 80 \\ \hline
\{1,2,3\} & 30 & 30 & 25 & 1 & 1 & 1 & 1 & 1 & 1 & 1 & 85 \\ \hline
\end{array}%
\end{equation*}

It is easy to check that the above game is concave, but condition $(iii)$ does
not hold. Indeed, for the first period, $\mathcal{E}^{1}=\left\{ 1\right\}$ but $y_{1}^{\ast }(\left\{ 2,3\right\})
=2<3= y_{1}^{\ast }(\left\{ 3\right\} ).$
\smallskip

Moreover, the Nucleolus $\eta (N,c)=\left( \frac{70}{3},
\frac{85}{3},\frac{100}{3}\right)$ is lightly different from the Shapley value,
$\phi (N,c)=\left( \frac{125}{6},\frac{155}{6}, \frac{115}{3}\right)$.
\smallskip

Finally, $o=(30,30,25)$, $\omega
=(25,30,30)$ and so the Solomonic allocation is $\varsigma (N,c)=\left( \frac{55}{2},30,\frac{55}{2}\right).$
\end{example}

\section{Concluding remarks}

This paper completes the study of the PI-games presented in Guardiola et al.
(2008, 2009). Those two papers proposed the Owen point as a natural core-allocation,
which does not pay attention to the role that essential players play in reducing the costs of their fans.
In that sense, essential players could consider the Owen point as an altruistic core-allocation.
However, the core was not studied in depth there.

Here we have analyzed carefully the core structure of PI-games,
and we have realized that the number of extreme point of its core is exponential
in the number of players. Then, we have proposed a new core-allocation, the Omega
point, that compensates the essential players for their role in reducing the
costs of their fans. Based on the Owen and Omega points we have defined the QPQ-set.
Since every QPQ allocation is a convex combination of the Owen and the Omega points,
we have paid special attention to the equally weighted QPQ allocation, the Solomonic allocation. Finally, we have provided some necessary conditions for the coincidence of the latter with the
Shapley value and the Nucleolus.


\begin{thebibliography}{99}
\bibitem{B63} {Bondareva ON (1963) Some applications of linear programming
methods to the theory of cooperative games.\ Problemy Kibernety 10:119-139}


\bibitem{DIN99} {Deng X, Ibaraki T, Nagamochi H (1999) Algorithmic aspect of
the core of combinatorial optimization games.\ Mathematics of Operations
Research 24:751-766}


\bibitem{routing97} {Derks, J. and Kuipers, J (1997) On the core of routing
games.\ International Journal of Game Theory 26:193-205}


\bibitem{faigle97} {Faigle U, Kern W, Fekete SP, Hochst\"{a}ttler W (1997)
On the complexity of testing membership in the core of min-cost spanning
tree games. International Journal of Game Theory 26:361-366}

\bibitem{fang2002} {Fang Q, Zhu S, Cai M and Deng X (2002) On the
computational complexity of membership test in flow games and linear
production games. International Journal of Game Theory 31:39-45}

\bibitem{SCH79} Schmeidler, D. 1969.\textbf{\ } The
Nucleolus of a Characteristic Funtion Game, \emph{\ }SIAM
Journal of Applied Mathematics 17, 1163-1170.


\bibitem{skutella2004} {Goemans M and Skutella M (2004) Cooperative facility
location games.\ Journal of Algorithms 50:194-214}


\bibitem{GMJ08} Guardiola LA, Meca A, Puerto J (2008) Production-inventory
games and PMAS-games: Characterizations of the Owen point. Mathematical
Social Sciences 56:96-108

\bibitem{GMJ09} {Guardiola LA, Meca A, Puerto J (2009) Production-Inventory
games: a new class of totally balanced combinatorial optimization games.\
Games and Economic Behavior 65:205-219}

\bibitem{HKSTV02} {Hamers H, Klijn F, Solymosi T, Tijs SH, Villar JP (2002)
\ Assignment games satisfy the CoMa-property.\ Games and Economic Behavior
38:231-239}

\bibitem{PPF12} Perea F, Puerto J, Fern\'{a}ndez FR (2012) Avoiding
unfairness of Owen allocations in linear production processes. European
Journal of Operational Research 220:125-131

\bibitem{KMM09} Kar A, Mitra M, Mutuswami S (2009) On the coincidence of the
prenucleolus and the Shapley value. {Mathematical Social Sciences 57:16-25}

\bibitem{kuiper} {Kuipers J (1993) On the Core of information graph games.
International Journal of Game Theory 21:339-350}










\bibitem{SLG01} {S\'{a}nchez-Soriano J, L\'{o}pez MA, Garc\'{\i}a-Jurado I
(2001)\ On the core of transportation games. \ Mathematical Social Sciences
41:215-225}

\bibitem{SH63} {Shapley LS (1953)\ A value for n-person games in Contributions to the Theory of Games II .\ Annals of Mathematics Studies. 28:307-317}

\bibitem{SH67} {Shapley LS (1967)\ On Balanced Sets and Cores.\ Naval Res.
Logist. 14:453-460}

\bibitem{SH71} {Shapley LS (1971)\ Cores of convex games. \ International
Journal of Game Theory 1:11-26}

\bibitem{SS69} {Shapley LS, Shubik M (1969)\ On market games.\ Journal of
Economics Theory 1:9-25}

\bibitem{sotomayor2003} {Sotomayor M (2003) Some further remarks on the core
structure of the assignment game. Mathematical Social Sciences 46:261-265}

\bibitem{S90} {Sprumont Y (1990) Population Monotonic Allocation Schemes for
Cooperative Games with Transferable Utility.\ Games and Economic Behavior
2:378-394}
\end{thebibliography}
\end{document}